\documentclass[authoryear,12pt,3p,letter]{jowarticle}
\usepackage{graphicx}
\usepackage{amsthm,amsmath}
\usepackage{amssymb}
\usepackage{amsfonts}
\usepackage{natbib}
\usepackage{setspace}
\usepackage[all]{xy}
\usepackage{enumitem}
\usepackage{titlesec}
\usepackage{mathrsfs}
\makeatletter
\renewcommand\section{\@startsection{section}{1}{\z@}{-3.25ex plus -1ex minus -.2ex}{1.5ex plus .2ex}{\normalsize\bf}}
\renewcommand\subsection{\@startsection{subsection}{2}{\z@}{-3.25ex plus -1ex minus -.2ex}{1.5ex plus .2ex}{\normalsize\bf}}
\renewcommand\subsubsection{\@startsection{subsubsection}{3}{\z@}{-3.25ex plus -1ex minus -.2ex}{1.5ex plus .2ex}{\normalsize\bf}}
\makeatother

\newtheorem{innercustomgeneric}{\customgenericname}
\providecommand{\customgenericname}{}
\newcommand{\newcustomtheorem}[2]{%
  \newenvironment{#1}[1]
  {%
   \renewcommand\customgenericname{#2}%
   \renewcommand\theinnercustomgeneric{##1}%
   \innercustomgeneric
  }
  {\endinnercustomgeneric}
}

\newtheorem{fact}{Fact}
\newtheorem{thm}{Theorem}

\newtheorem{prop}[thm]{Proposition}

\newcustomtheorem{prop*}{Proposition}
\newcustomtheorem{cor*}{Corollary}

\begin{document}
\begin{frontmatter}
\title{Are General Relativity and Teleparallel Gravity\\ Theoretically Equivalent?}
\author{James Owen Weatherall}\ead{weatherj@uci.edu}
\address{Department of Logic and Philosophy of Science\\ University of California, Irvine}
\author{Helen Meskhidze}\ead{emeskhidze@g.harvard.edu}
\address{Black Hole Initiative \\ Harvard University}
\begin{abstract}
Teleparallel gravity shares many qualitative features with general relativity, but differs from it in the following way: whereas in general relativity, gravitation is a manifestation of spacetime curvature, in teleparallel gravity, spacetime is (always) flat. Gravitational effects in this theory arise due to spacetime torsion.  It is often claimed that teleparallel gravity is an equivalent reformulation of general relativity.  In this paper we question that view.  We argue that the theories are not equivalent, by the criterion of categorical equivalence and any stronger criterion, and that teleparallel gravity posits strictly more structure than general relativity.
\end{abstract}
\end{frontmatter}
\doublespacing
\section{Introduction}\label{sec:intro}

Teleparallel Gravity (TPG) is a theory of gravitation that is (1) relativistic, i.e., its metrical structure is Lorentzian; (2) geometrized, i.e., the metric is dynamically related to stress-energy; and (3) empirically equivalent to general relativity (GR), at least locally \citep{TPGBook,Maluf}.  But whereas in GR, space-time is curved in the presence of matter, in TPG, space-time is always flat, as in (non-geometrized) Newtonian gravitation.  The characteristic phenomena of GR arise in TPG from space-time torsion, rather than space-time curvature.  Thus TPG puts pressure on the idea that spatiotemporal curvature is necessary to account for relativistic phenomena.  A physicist (or philosopher) motivated to work in flat space-time, whether for pragmatic reasons or metaphysical ones, can apparently adopt TPG over GR without risk of violating empirical constraints.

Another topic of recent interest concerns whether TPG, as a relativistic theory set in flat spacetime, bears the same relationship to GR that Newtonian gravitation does to geometrized Newtonian gravitation (aka Newton-Cartan theory), which is a formulation of Newtonian gravitation in which spacetime is curved and gravitation is a manifestation of that curvature. \citet{Read+TehCQG} have argued that the relationship between GR and TPG is analogous to that between geometrized and standard Newtonian gravitation, and indeed, that one can think of ordinary Newtonian gravitation as the ``teleparallelization'' of geometrized Newtonian gravitation.  \citet{Meskhidze+Weatherall}, meanwhile, have argued that the Newtonian theory most strongly analogous to teleparallel gravity is a strict generalization of ordinary Newtonian gravitation, where the Newtonian gravitational force is replaced by a force field related to the torsion of a flat derivative operator.  Curiously, though, while this torsional Newtonian theory is plausibly the classical analogue to teleparallel gravity, \citet{MeskhidzeDiss} shows that the non-relativistic limit of teleparallel gravity is, in fact, ordinary Newtonian gravitation, without torsion.\footnote{\citet{Read+TehCQG} make a superficially similar claim about the reduction of TPG to Newtonian gravitation, but they do not consider the non-relativistic limit; instead, they consider what is known as ``dimensional reduction'', which concerns the geometry induced on certain four dimensional surfaced embedded in five dimensional models of general relativity. \citet{Schwartz} also considers a non-relativistic limit of TPG.}

It is broadly taken for granted within the physics literature that GR and TPG are not only empirically equivalent, but theoretically equivalent as well.  \citet{TPGBook}, for instance, write “Teleparallel Gravity … is an alternative description that, though equivalent to General Relativity, separates gravitation from inertia… (vii).” In other literature, Teleparallel Gravity is referred to as the “Teleparallel Equivalent to General Relativity (TEGR)” \citep[c.f.][]{Maluf}.  Some authors place significant conceptual weight on this equivalence.  For instance, \citet{TPGBook} suggest a thesis concerning the conventionality of geometry on the basis of the equivalence of these theories, arguing that both curvature and torsion are equally good ways of representing the same gravitational degrees of freedom.  They also suggest that the equivalence of these two theories reveals something deep about the universality of gravitation, which they emphasize as the feature that distinguishes gravity from other fundamental interactions.  They write, ``As the sole universal interaction, it is the only one to allow also a geometrical interpretation—hence the possibility of two descriptions. From this point of view, curvature and torsion are simply alternative ways of representing the same gravitational field, accounting for the same gravitational degrees of freedom" (p. ii).

Philosophers, too, to the extent they have discussed these issues at all, have largely endorsed some version of the equivalence claim.  \citet{Read+TehCQG}, for instance, call Newtonian gravitation the ``teleparallel equivalent'' of Newton-Cartan theory, and use similar language to discuss TPG and GR.  \citet{Knox}, meanwhile, entertains the possibility that in virtue of positing flat spacetime, TPG could be inequivalent to GR, but she concludes that when properly interpreted the theories are not distinct.\footnote{To be sure, she argues the theories are equivalent because, properly understood, they }

What philosophers have not yet done is bring the formal machinery for assessing theoretical equivalence developed in the philosophical literature over the past decade or so, building on earlier work by \citet{GlymourTETR}, \citet{Quine}, and others, to analyze the relationship between TPG and GR.\footnote{See \citet{WeatherallTE1,WeatherallTE2} for a recent review of the theoretical equivalence literature.}  That is the task we take up in the present paper. Doing so will both clarify the relationship between these two theories by drawing out the features that past work has suggested might be salient to assessing their (in)equivalence, and it will allow for clearer comparisons with other, better known, candidates for equivalent theories.  It will also clarify some aspects of the theories themselves.  As we will show, there is a straightforward sense, known but rarely emphasised within the TPG literature, in which TPG and GR fail to be equivalent.  This is because the natural, empirical-content-preserving map associating models of TPG to models of GR is many-to-one, and thus fails to meet the necessary condition for (definitional) equivalence first identified by \citet{GlymourTETR}.  In this sense, then, TPG and GR turn out to be inequivalent in just the sense that Glymour argues ordinary Newtonian gravitation is inequivalent to geometrized Newtonian gravitation.  In fact even more is true: this many-to-one relationship holds even up to model isomorphism, so that the theories are also inequivalent by the weaker criterion of categorial equivalence.\footnote{Note that this observation also implies that the theories fail to be Morita equivalent \citet{Barrett+Halvorson}.}   

One might attempt to recover equivalence by modifying TPG, analogously to how \citet{WeatherallNGE} suggests one can modify Newtonian gravitation as traditionally conceived, by introducing a new notion of model equivalence that renders (distinct, non-isomorphic) models of TPG associated with the same model of GR equivalent.  But this move is unattractive, because it will turn out that the shared structure among such models is precisely the shared spacetime metric, and thus it appears that the resulting theory is \emph{identical} to GR.  Conversely, if one does not introduce these additional model equivalences -- as one must not to maintain that TPG is a theory of \emph{flat} space-time -- TPG turns out to posit surplus structure relative to GR, in the sense of \citet{WeatherallUG}.  There are distinct models, by the light of the theory, that agree on all of their empirical content.  The theory that results from excising that structure is, again, GR.

The remainder of the paper will proceed as follows.  We will begin, in section \ref{sec:equivalence}, by reviewing some relevant background from the theoretical equivalence literature.  Then, in section \ref{sec:prelim}, we will provide some mathematical preliminaries, with an eye towards fixing necessary facts about Lorentzian geometry in the presence of torsion.  In section \ref{sec:TPG+GR}, we will introduce GR and TPG in more detail.  We then turn to the main arguments of the paper.  In section \ref{sec:main}, we will establish three elementary lemmas that underlie our argument that these theories are inequivalent.  

\section{Theoretical Equivalence \& Structure}\label{sec:equivalence}

Over the past five decades, philosophers have discussed a variety of formal criteria of equivalence of physical theories.\footnote{See \citet{WeatherallTE1,WeatherallTE2} for a recent review.  There is also an established tradition of rejecting formal criteria of equivalence: see, for instance, \citet{Sklar}, \citet{Coffey}, \citet{Nguyen}, and \citet{Teitel}.  Others, such as \citet{NorthStructure,NorthBook}, have used related formal methods to assess (in)equivalence, without providing a once-and-for-all criterion.}  In now-classic work, for instance, \citet{GlymourTETR,GlymourTE} proposed definitional equivalence as a plausible criterion for theoretical equivalence of scientific theories.  Roughly speaking, definitional equivalence obtains of two (first-order) theories (with disjoint signatures) if one can extend both theories by adding suitable ``definition'' formulas in such a way that the theories that result are logically equivalent.  In other words, definitionally equivalent theories ``say the same things'' in different terms. Glymour argued that adopting this criterion implied certain necessary conditions for equivalence that could be used to show that certain \emph{empirically} equivalent theories are nonetheless theoretically \emph{in}equivalent.  

In more recent work, \citet{WeatherallNGE}, inspired by arguments due to \citet{Halvorson}, proposed categorical equivalence as an alternative criterion that could make better sense of certain common equivalence judgments in physics.\footnote{More recently still, \citet{Hudetz} has criticized categorical equivalence and has proposed a refinement that strengthens it; and \citet{WeatherallWNCE} has expressed reservations about categorical equivalence, including Hudetz's proposal.  \citet{MarchKPM} has recently proposed a promising way forward.  Even so, for reasons explained below, the issues Hudetz and Weatherall identify will not trouble us here.}  \citet{Barrett+HalvorsonME}, meanwhile, have proposed a different weakening of definitional equivalence that they call ``Morita equivalence''. Morita equivalence is like definitional equivalence in that it involves introducing a translation manual between extended theories; but it is more general, in the sense that it allows one to define new sorts, rather than just function and predicate symbols.  As \citet{Barrett+HalvorsonME} show, any two theories that are definitionally equivalent are Morita equivalent; and any two theories that are Morita equivalent are categorically equivalent.  The converses, however, do not hold.\footnote{The relationship between these notions of equivalence and several other relations between theories is exhaustively explored by \citet{Meadows}.}  

In virtually all work on these criteria of equivalence, two theories are said to be theoretically equivalent (if and) only if they are empirically equivalent and, in addition, satisfy a further formal criterion of equivalence.\footnote{See \citet{WeatherallTE1} for a discussion of the challenges arising from requiring empirical equivalence.  Note, too, that in this sense these criteria are not ``purely'' formal, \emph{pace} some critics.}  Moreover, one requires for equivalence that the sense in which two theories are empirically equivalent and the sense in which they are formally equivalent are compatible. The requirement of (compatible) empirical equivalence, which we explain in more detail below, is imposed to ensure that formally equivalent theories with different subject matters are not judged to be equivalent.  Thus, empirical equivalence is understood to be the weakest criterion under consideration, since it is required, by stipulation, for all of these criteria.  Of the criteria widely considered in the literature, categorical equivalence is the next weakest, followed by Morita equivalence, definitional equivalence, and logical equivalence, which is the strongest.\footnote{Other criteria, such as KPM-categorical equivalence \citep{MarchKPM}, require somewhat more nuance to apply, and so they do not fit clearly within this hierarchy---though they are certainly not \textit{weaker} than categorical equivalence.  To our knowledge, though, there are no plausible criteria in the literature that lie between categorical equivalence and empirical equivalence.}

We will illustrate how the compatibility requirement works in more detail by stating categorical equivalence as a criterion of theoretical equivalence.\footnote{The foregoing will take for granted some basic concepts from category theory.  For background on that theory, see (for instance) \citet{Leinster}.}  (This is the only criterion that will concern us in what follows, since it is the weakest criterion of theoretical equivalence.)  Suppose we have two physical theories, Theory 1 and Theory 2.  We first identify categories that we use to represent these theories for the purpose of assessing their equivalence.  Typically, these will be categories $\mathcal{T}_1$ and $\mathcal{T}_2$ whose objects are models of the respective theories and whose arrows are ``structure preserving maps'' between those models.  In many cases, authors focus on groupoids of models, which are categories whose arrows are (all) isomorphisms, though in other cases, such as the equivalence (actually, duality) result between Einstein Algebras and General Relativity, \citet{Rosenstock+etal} consider a broader class of maps that preserve metrical structure.  Considerable judgement goes into constructing these categories, which introduces a second sense (beyond empirical equivalence) in which categorical equivalence is not ``purely'' formal.  Although sometimes more than one choice can be motivated, those choices generally also track interpretive choices, and in that sense are informative \citep[c.f.][]{WeatherallNGE,BarrettLH2,MarchMG}.

In addition, we fix ``functors'' $E_1:\mathcal{T}_1\rightarrow \mathcal{W}$ and $E_2:\mathcal{T}_2\rightarrow \mathcal{W}$, relating these models to the ``predictions'' of the theories.  Here $\mathcal{W}$ is some characterization, at a level that is neutral between the theories under consideration, of ``observable states of the world'', with some notion of equivalence between those states.  Precisely how to think about and define $\mathcal{W}$ will be highly context specific.  Moreover, in general one should expect theories that admit of such functors (or some shared category $\mathcal{W}$ that can serve as the codomain of those functors) to already be very similar. (These ideas are abstract at present, but will be made more concrete in the example we discuss below.) 

Having specified this data, we then say that Theory 1 and Theory 2 are \emph{categorically equivalent} if there exist functors $F_1:\mathcal{T}_1\rightarrow\mathcal{T}_2$ and $F_2:\mathcal{T}_2\rightarrow\mathcal{T}_1$ such that:
\begin{enumerate}
\item $F_1$ and $F_2$ are essentially inverses, i.e., they realize an equivalence of categories; and
\item $E_2\circ F_1 = E_1$ and $E_1\circ F_2 = E_2$.
\end{enumerate}
 The first condition states that the two categories are equivalent as categories, whereas the second states that they are empirically equivalent in a way that is compatible with that categorical equivalence.  That is, models of theory 1 are mapped to models of theory 2 with the same empirical consequences, and vice versa.

One advantage of categorical equivalence is that it provides tools for interpreting \emph{inequivalence} of sufficiently similar theories.  In particular, suppose that one has two theories with associated categories $\mathcal{T}_1$ and $\mathcal{T}_2$ as above and suppose there exists a functor $F_1:\mathcal{T}1\rightarrow \mathcal{T}_2$ such that $E_2\circ F_1 = E_1$.  But now suppose that $F_1$ is \emph{not} essentially invertible.  Then there is a sense in which theory 2 can be said to reproduce the empirical content of theory 1, but by that standard of comparison, the theories are not equivalent.  In that case, properties of $F_1$ can be helpful guides for thinking about the relationship between the theories.  

Following arguments due to \citet{Baez+etal}, \citet{WeatherallUG}, and \citet{BarrettPSS}, we say that theory 1 posits more \emph{structure} than theory 2 if $F_1$ fails to be full (its induced map on hom sets is not surjective); theory 1 posits more \emph{stuff} than theory 2 if $F_1$ fails to be faithful (its induced map on hom sets is not injective); and theory 1 assumes more \emph{properties} than theory 2 (or theory 1 \emph{generalizes} theory 2) if $F_1$ fails to be essentially surjective (it is not surjective on objects, even up to isomorphism).\footnote{This framework was originally introduced to philosophy of science to address structural comparisons, and so failures of fullness were most salient \citep{WeatherallNGE,WeatherallUG}.  For a discussion of what it means to forget stuff, see \citet{Nguyen+etal} and \citet{Bradley+Weatherall}.}  Note that a functor can simultaneously fail to be full, faithful, and essentially surjective in any combination (and that if $F_1$ is full, faithful, and essentially surjective, then it is essentially invertible).

\section{Geometrical Background}\label{sec:prelim}

We now turn to some mathematical preliminaries needed to describe GR and TPG in a uniform geometric framework.  Since there is some variation in formalism, definitions, and conventions between different parts of the TPG literature and the GR literature, which sometimes confounds straightforward comparison, our goal here is to provide a uniform setting in which the two theories can be presented systematically.  (We assume familiarity with standard GR.\footnote{For more detail on GR using the same conventions that we do, see \citet{Malament}; \citet{Wald} uses similar conventions, though note that Wald and Malament differ by a sign in their metric signature, which leads to several other sign changes in formulas for, e.g., curvature, torsion, and related concepts.})  Note that in what follows, we will assume that manifolds are smooth, paracompact, Hausdorff, without-boundary, and parallelizable. Parallelizability is not typically required for models of GR, but it is required for ``global'' models of TPG.\footnote{A manifold is said to be \textit{parallelizable} if its tangent bundle is trivial, or, equivalently, if it admits a global framefield. Is parallelizability an unreasonable restriction, or one that already shows GR to be more general than TPG?  Perhaps.  But on the other hand, parallelizability is required for a manifold to admit a (global) spinor structure, and so insofar as the world admits matter fields with spinorial degrees of freedom, one might think even GR must be restricted to the parallelizable sector.}  All of the fields and other structures we define on manifolds will be assumed to be smooth.  

We will need somewhat generalized versions of (GR) textbook definitions to accommodate TPG. First, recall that a \emph{covariant derivative operator} $\nabla$ on a manifold $M$ is a linear map from rank $(r,s)$ tensor fields to rank $(r,s+1)$ tensor fields, for any $r,s\geq 0$, satisfying the Leibniz rule and agreeing with the gradient operator on scalar fields: $\nabla_a\alpha = d_a\alpha$ for every scalar field $\alpha$. A covariant derivative operator determines a standard of \emph{parallel transport} for vector fields along curves.  Note that unlike in some definitions, we do not require that for scalar fields $\alpha$, $\nabla_{[a}\nabla_{b]} \alpha = \mathbf{0}$.  Instead, we define the \emph{torsion} of a derivative operator $\nabla$ to be the tensor field $T^a{}_{bc}$ such that for any scalar field $\alpha$,\footnote{\citet{WaldTorsion} extends Wald's treatment of derivative operators systematically to include torsion.  Our treatment follows that one closely, except that we use Malament's sign conventions.}
\[
2\nabla_{[a}\nabla_{b]}\alpha = T^n{}_{ab}\nabla_n\alpha.
\]
(That torsion is a tensor at all is a non-trivial fact, whose proof is similar to proofs that curvature is a tensor.)  Torsion is a measure of how a basis for the tangent space ``twists'' when parallel-transported along a curve, relative to $\nabla$. A derivative operator is \emph{torsion-free} if its torsion vanishes, in which case the antisymmetrized derivative of a scalar field vanishes.  In GR, the derivative operator is always assumed to be torsion-free.
 
In the presence of torsion, the \emph{curvature} of a derivative operator $\nabla$ is defined as the tensor field $R^a{}_{bcd}$ such that for any vector field $\xi^a$,
\[
R^a{}_{bcd}\xi^b = -2\nabla_{[c}\nabla_{d]}\xi^a + T^n{}_{cd}\nabla_n\xi^a
\]
(Again, to show that such a tensor exists is not trivial, but it is a standard result.  Note, too, that it is only the \textit{sum} of the two terms on the right that defines a tensor; the individual summands are not generally tensors in the presence of torsion.)  A derivative operator is \emph{flat} if its curvature vanishes. Curvature is usually understood as a measure of the path-dependence of parallel transport relative to a derivative operator; this interpretation does not change with the introduction of torsion, and in particular, parallel transport relative to $\nabla$ is locally path-independent just in case $\nabla$ is flat.

Now let $g_{ab}$ be a (non-degenerate) metric on a manifold $M$, of any signature.  A derivative operator on $M$ is \emph{$g$-compatible} if $\nabla_ag_{bc}=\mathbf{0}$.  The fundamental theorem of Riemannian geometry, generalized to pseudo-Riemannian metrics, states that there always exists a unique torsion-free $g$-compatible derivative operator $\nabla$ on $M$.  This derivative operator is known as the Levi-Civita derivative operator; in general, the Levi-Civita derivative operator will have non-vanishing curvature. Once we allow torsion, however, there are always many $g-$compatible derivative operators available.  In fact, given any smooth tensor $T^a{}_{bc}$, anti-symmetric in its lower indices, there always exists a (unique) $g-$compatible derivative operator with torsion $T^a{}_{bc}$.  In general, these torsionful derivative operators will have curvature.  However, if one allows arbitrary torsion, one can also always find $g-$compatible \emph{flat} derivative operators.  

Note that there is an asymmetry between curvature and torsion here.  If we stipulate that torsion vanishes, there is always a unique (curved) $g-$compatible derivative operator; if we stipulate that curvature vanishes, there are always \emph{many} torsionful derivative operators available.  How we do know this?  One can construct flat, $g-$compatible derivative operators by considering orthonormal frames.  Any such frame will determine a unique (generally torsionful) derivative operator relative to which that frame is parallel (i.e., constant).\footnote{Note that this frame will in general be \textit{anholonomic}, which means it is not associated with a coordinate chart.  In fact, a flat derivative operator has torsion if and only if its constant frames are anholonomic.}  This derivative operator will be $g-$compatible due to the fact that the frame is orthonormal; and it will be flat because it admits a basis of constant fields (namely, the frame).  Now consider a change of frame that acts as a Lorentz transformation at each point, where that transformation is not constant with respect to the induced derivative operator.  The result will be a new orthonormal frame, which will induce a distinct $g-$compatible, torsionful, flat derivative operator.\footnote{The argument here may seem at odds with claims in the TPG literature that torsionful derivative operators are invariant under ``local Lorentz transformations'' \cite[c.f.][\S 5.2.1]{TPGBook}.  What is true is that any connection, with torsion or without, can be expressed relative to any frame. In other words, given a connection, expressed relative to one frame, one can always express the same connection relative to a different frame, related to the first by a local Lorentz transformation.  But in general, at most one of the frames will be constant relative to the connection; and the two connections determined by the two frames will be distinct (and both $g-$compatible). 
 We are grateful to James Read for noting this apparent conflict.}

Suppose that $\nabla$ is the unique torsion-free derivative operator compatible with a metric $g_{ab}$.  Then given any other derivative operator $\tilde{\nabla}$ compatible with $g_{ab}$, the \emph{contorsion} tensor is the field $K^a{}_{bc}$ such that for any vector field $\xi^a$,
\[
\tilde{\nabla}_a\xi^b = \nabla_a\xi^b - K^b{}_{an}\xi^n
\]
The contorsion tensor can be defined directly in terms of the torsion, as 
\[
K^a{}_{bc}=\frac{1}{2}\left(T^a{}_{bc}-T_{cb}{}^a-T_{bc}{}^a\right)
\]
Note that the contorsion tensor is neither symmetric nor anti-symmetric in its lower indices.  However, $2K^a{}_{[bc]} = T^a{}_{bc}$, where $T^a{}_{bc}$ is the torsion of $\tilde{\nabla}$. 

Finally, we consider a notion of structure-preserving maps for manifolds with (possibly torsionful) derivative operators.  Let $M$ and $M'$ be manifolds of the same dimension, with derivative operators $\nabla$ and $\nabla'$ respectively. An embedding $\varphi:M\rightarrow M'$ \emph{preserves $\nabla$} if for all smooth vector fields $\xi^a$ on $\varphi[M]\subseteq M'$, $\varphi_*(\nabla_a\varphi^*(\xi^b))=\nabla'_a\xi^b$.  Similarly, if $M$ and $M'$ are endowed with metrics $g_{ab}$ and $g'_{ab}$, respectively, we say an embedding is \emph{isometric} if $\varphi^*(g'_{ab})=g_{ab}$.\footnote{Recall that an embedding is always a diffeomorphism onto its image.}  Note that because the Levi-Civita derivative operator is uniquely determined by a metric, any isometric embedding between manifolds with metrics preserves the Levi-Civita derivative operator.  But the analogous result does not hold for $g-$compatible derivative operators with torsion.  That is, if $\nabla$ and $\nabla'$ on $M$ and $M'$ are compatible with $g_{ab}$ and $g'_{ab}$, respectively, and $\varphi$ is an isometric embedding, it does not follow that $\varphi$ preserves $\nabla$.  This is closely related to the non-uniqueness of torsionful $g-$compatible derivatives noted above.  

\section{TPG \& GR}\label{sec:TPG+GR}

We now turn to defining the theories under consideration.  The theories share a great deal of structure, and admit similar interpretations.  In both cases, we define the theories by specifying a class of models, i.e., mathematical structures we use to represent certain classes of physical situations.  The sorts of situations we are interested in representing with these models, in both cases, are possible (regions of) universes.  Both theories begin with a four-dimensional manifold, $M$, assumed, as noted above, to be smooth, Hausdorff, paracompact, and parallelizable.  Points of this manifold represent events, localized in space and time.  In addition, in both theories we assume this manifold is endowed with a Lorentz-signature metric, $g_{ab}$, which determines spatiotemporal relations between events and, in particular, causal structure.  The interpretation of this metric is the same in both theories, in the sense that it determines local lengths of curves, and it distinguishes between spacelike, timelike, and null vectors.  All of this is just as in standard GR, and we will not dwell on it.

The key differences in the theories comes in their affine structure, as determined by a derivative operator.  In GR, we add to this structure the Levi-Civita derivative operator determined by $g_{ab}$.    (In many treatments of GR, this structure is suppressed, since it is uniquely determined by $g_{ab}$, but there is no harm in mentioning it explicitly.) This derivative operator determines a standard of parallel transport, and, in particular, a class of geodesics.  The timelike geodesics represent the ``force-free'' motions of small massive point particles; while the null geodesics represent the possible trajectories of light rays.  Forces act on bodies by accelerating them, relative to this derivative operator.  And finally, associated with any matter in spacetime is a rank-2 tensor field $T^{ab}$, representing the energy and momentum properties of that matter.  The total energy-momentum tensor, that is, the sum of energy-momentum tensors associated with all matter fields present in space and time, is required to be divergence-free with respect to the Levi-Civita derivative operator.\footnote{For details on the status of this final condition, see \citet{WeatherallCC}.}

For present purposes, we take the empirical content of a model of GR to be exhausted by the structure $(M,g_{ab})$.  This is justified because $g_{ab}$ uniquely determines $\nabla$; and $g_{ab}$ is determined, up to a constant multiple, by its timelike and null geodesics.\footnote{This result is original to \citet{Weyl}; see \citet[\S 2.1]{Malament} for a discussion.}  The picture behind this choice is that we probe the structure of spacetime by studying the motions of free massive point particles and light rays.  Thus the physical significance of the metric is given by the possible motions of those bodies.  This is a highly idealized picture.  One could of course enrich this conception of empirical content, to specify other kinds of observable matter fields.  But we suggest that any more complicated picture of this sort will ultimately involve using the dynamical behavior of that matter, in a limit where it only weakly perturbs the background spacetime structure; and that that dynamical behavior will be determined by the metric (and other structures uniquely determined thereby).  So, we would argue, nothing in what follows depends on adopting the idealized choice we do here.
 
We now turn to TPG.  Here, too, we begin with a pair $(M,g_{ab})$ as above, and add to it a $g-$compatible derivative operator $\tilde{\nabla}$.  Now, however, we require $\tilde{\nabla}$ to be flat, but allow it to have torsion.  The interpretation of this derivative operator is essentially the same as above, in the sense that it determines a class of geodesics, and one takes the timelike and null geodesics to be possible trajectories for small, massive bodies and light rays, respectively, in the absence of any gravitational interactions.  Importantly, however, gravitation is now viewed as a force, which induces an acceleration for massive bodies (and light rays).  This force is determined by the contorsion tensor $K^a{}_{bc}$, so that the trajectory of a massive point particle will be represented by a curve whose acceleration is given by
\begin{equation}\label{Force}
\xi^n\tilde{\nabla}_n\xi^a = -K^a{}_{nm}\xi^n\xi^m
\end{equation}
where $\xi^a$ is the 4-velocity of the particle.  Thus, though the derivative operator determines a class of force-free motions, the actual trajectories of bodies will generally differ from those geodesics.  Likewise, the energy-momentum associated with matter will not, in general, be divergence-free, relative to $\tilde{\nabla}$; instead, the divergence of the total matter energy-momentum $T^{ab}$ is given by:
\[
\tilde{\nabla}_a T^{ab} = -K^a{}_{an}T^{nb} - K^b{}_{an}T^{an}.
\]
The interpretation of this equation is that the total energy-momentum of matter changes, locally, due to the effects of gravitation.  

We once again take the empirical content of a model to be exhausted by the structure $(M,g_{ab})$.  This is because, although the distinction between force-free and forced motion in these models depends on the derivative operator $\tilde{\nabla}$, the trajectories that massive bodies and light rays follow is in fact entirely determined by the metric.  In fact, Eq. \eqref{Force} implies that the trajectories of bodies will be the geodesics of the Levi-Civita derivative operator, just as in general relativity.

Note that we have said nothing, here, about gravitational dynamics. Of course, in both theories, one can write down dynamical equations relating properties of spacetime -- in GR, the metric and spacetime curvature (viz., Einstein's equation); in TPG, the metric and spacetime torsion -- to the distribution of energy-momentum in space and time.  But the details of how this works will be irrelevant in what follows, and so we neglect it. Instead, we assume all models, in the sense described above, are dynamically possible for some energy-momentum distribution or other (allowing for physically unrealistic matter, which may not arise from any known matter fields, satisfy energy conditions, etc.) Moreover, we simply stipulate that the theories are dynamically equivalent, in the sense that if we have a dynamically possible model $(M,g_{ab},\nabla)$ of one of the theories, for some energy-momentum tensor $T^{ab}$, then there will be a corresponding model $(M,g_{ab},\tilde{\nabla})$ of the other theory, with the same metric, which is dynamically possible with the same energy-momentum tensor.  Our justification for this is both that it appears to be widely accepted in the literature and, more importantly, if it failed, the theories would be empirically inequivalent, and therefore theoretically inequivalent on all of the criteria considered above.  

\section{Principal Claims, Defended}\label{sec:main}

With this background in place, we now turn to the paper's principal claims: GR and TPG are not equivalent, and GR has \emph{less structure} than TPG.  In fact, this follows directly from observations already made; one only needs to put them together and fit them into the equivalence framework described above.  We rely on the following facts, all established previously:
\begin{fact} For every model of TPG, there exists a unique model of GR with the same metric. 
\end{fact}
\begin{fact} For every model of GR, there exist many distinct models of TPG with the same metric. 
\end{fact}
\begin{fact} In general, an isometric embedding preserves a metric-compatible derivative operator only if it is torsion-free.
\end{fact}
 
Now we assemble the pieces.  First, we define categories associated with each of our theories.  Let $\mathcal{GR}$ be the category whose objects are models $(M,g_{ab},\nabla)$ of GR and whose arrows are isometric embeddings. And let $\mathcal{TPG}$ be the category whose objects are models $(M,g_{ab},\tilde{\nabla})$ of TPG and whose arrows are isometric embeddings that preserve the derivative operator.  Finally, for both theories, define the ``world'' (of empirical consequences) as the category whose objects are pairs $(M,g_{ab})$ and whose arrows are isometric embeddings.  Empirical consequence functors $E_{GR}$ and $E_{TPG}$ are defined in the obvious way, as taking models to their reducts $(M,g_{ab})$.
 
Now observe that there is a functor $F:\mathcal{TPG}\rightarrow\mathcal{GR}$ that commutes with $E_{GR}$ and $E_{TPG}$.  $F$ acts on objects $(M,g_{ab},\tilde{\nabla})$ by taking them to $(M,g_{ab},\nabla)$, where $\nabla$ is the Levi-Civita derivative operator.  It takes arrows to themselves. We then have the following result:
\begin{prop} $F$ is not an equivalence of categories.  It is faithful and essentially surjective, but it is not full.
\end{prop}
\begin{proof} To show $F$ is not an equivalence, it is sufficient to show that it is not full.  But this follows from Fact 3 above, because every isometric embedding preserves the Levi-Civita derivative, but there are isometric embeddings that do not preserve torsionful $g-$compatible metrics.  These will be arrows in $\mathcal{GR}$ that do not lie in the image of $F$.  Alternatively, we can see this result following from Facts 1 and 2 in conjunction, because those facts establish that $F$ maps distinct, non-isomorphic objects $(M,g_{ab},\tilde{\nabla})$ and $(M,g_{ab},\tilde{\nabla}')$ to a unique object $(M, g_{ab},\nabla)$.  The hom-set between these models is empty, and so no arrow is mapped to the identity on $(M, g_{ab},\nabla)$ under this action. \end{proof}

This result shows that $F$ does not realize an equivalence.  But perhaps $F$ is not the right functor to consider.  One might wonder, for instance, if there is another functor that does realize an equivalence.  But recall that for categorical equivalence of physical theories, we demand that the functor must commute with the ``empirical'' functors $E_{GR}$ and $E_{TPG}$.  We claim that $F$ is the unique functor from $\mathcal{TPG}$ to $\mathcal{GR}$ that does this, and so it is the unique functor that could realize an equivalence (if the theories were equivalent at all).  This follows from Fact 1, since $F$ must preserve the metric in order to commute.

\section{Conclusion}

The foregoing establishes that TPG is not categorically (or Morita, or definitionally, etc.) equivalent to GR, because TPG posits more structure than GR.  Indeed, since the theories are empirically equivalent, one might conclude, following the arguments of \citet{WeatherallUG}, that TPG posits \textit{excess structure}, because both theories have the capacity to represent all of the same situations, but TPG does so in a redundant way.  The basic reason can be seen in two ways.  One way to see it is to note that the structure of models of TPG is preserved by fewer maps than the corresponding models of GR.  The other way to see it is to note that the natural, empirical-consequence-preserving mapping between models of TPG and models of GR is many-to-one, even up to isomorphism.
 
Now that everything has been laid out, one might think this result is surprising.  After all the models look like they have the same number and type of elements: $(M,g_{ab},\nabla)$ vs. $(M,g_{ab},\tilde{\nabla})$. Each of them attributes metrical structure and (compatible) affine structure to the world; they differ only in whether they require torsion or curvature to vanish.  But this comparison is misleading, and it shows why simple examination of the tuples we use to represent our models is not an adequate way to make structural comparisons of theories.  The difference is related to the ``implicit definability conception (IDC)'' of structure \citep{Barrett+etal}.  The crucial fact in the present case is that a metric implicitly defines its torsion-free derivative operator, in the sense that every smooth map that preserves the metric also preserves the Levi-Civita derivative operator.  But the metric does not implicitly define other $g-$compatible derivative operators. These other derivative operators require, or determine, further structure, namely, a class of preferred (anholonomic) frames.  

The case discussed here is analogous to other cases of inequivalence in the literature, such as the relationship between Newtonian gravitation and Newton-Cartan theory \citep{GlymourTE,WeatherallNGE} or the relationship between vector potential and electromagnetic field formulations of electromagnetism \citep{WeatherallNGE,WeatherallUG}.  In those cases, there is a sense in which the theories are inequivalent; but also a well-studied sense in which equivalence can be restored by introducing additional maps -- roughly speaking, \textit{gauge transformations} -- to equivocate between models of the theory with excess structure.   This strategy, called \textit{sophistication} by \citet{DewarSoph}, can be seen as implementing a reinterpretation of the models of the theory with more structure, where one stipulates that the structure that distinguishes the theories is to be neglected for the purposes of interpreting the physical significance of the models of the theories.

Could we pursue that strategy there?   Yes.  But the strategy trivializes the problem.  The question often arises, when one applies sophistication, whether there is an intrinsic formulation of the sophisticated theory available, in which the ``extra'' arrows are naturally seen as isomorphisms of that structure \citep{DewarSoph,Jacobs}.   In other cases, some work is needed to find such a formulation, and by doing that work, one often learns something about the theories in question.  For instance, in the Newtonian case, one finds that the additional arrows one adds turn out to be \emph{Maxwell transformations}, so that one can see the background spacetime structure of the ``sophisticated'' theory as Maxwell spacetime, rather than Galilean spacetime. This observation leads to the very interesting question of whether a satisfactory version of Newtonian gravitation can be formulated on Maxwellian spacetime \citep{Saunders,KnoxEP,WeatherallMaxwell,WeatherallRot,Dewar,Chen,MarchMG}.  But in the present case, something else happens.  If we equivocate between models associated with different metric-compatible derivative operators in TPG, we are left with just the structure of the manifold and metric, $(M,g_{ab})$.  This structure implicitly defines a torsion-free metric-compatible derivative operator---or in other words, it is just a model of GR.  Thus, there is a precise sense in which TPG, with its excess structure removed, is simply GR.

Where does this leave us?  We think that the arguments given here provide general reasons to think that GR provides a more accurate, or perspicuous, representation of reality than TPG.  To see the point, it is perhaps helpful to contrast our view with that of \citet{Knox}.  Knox argues that TPG, properly understood, describes the curved spacetime of general relativity in an imperspicuous way.  In more detail, she argues that, since small bodies in TPG follow the geodesics of the Levi-Civita derivative operator, the inertial structure of TPG is the same as the inertial structure of GR.  She then invokes to the strong equivalence principle to argue that a spacetime functionalist should attribute to TPG the same curved spacetime structure as in GR. We agree with this point---though we also note that the argument is only persuasive to those who accept the probative role that she grants the strong equivalence principle. 
We end up in a similar place -- that TPG, properly construed, is an imperspicuous characterization of GR -- but for different reasons.  On our accounting, it is that TPG has surplus structure that, when eliminated, reproduces GR that leads us to conclude that TPG is a redundant description of what is more perspicuously described by GR. 
 
\section*{Acknowledgments}
We are grateful to Eleanor March, Ruward Mulder, and James Read for discussion of this material, and to a helpful audience at the 2023 Foundations of Physics meeting in Bristol, UK for questions and comments.

\singlespacing

\bibliography{teleparallel}
\bibliographystyle{elsarticle-harv}

\end{document}